\documentclass[runningheads]{llncs}

\usepackage{amsthm}
\usepackage{graphicx}
\usepackage{amssymb}

\usepackage{stmaryrd}

\usepackage{tikz}
\usetikzlibrary{arrows,positioning,shapes,decorations,automata,backgrounds,petri}

\newcommand{\commentout}[1]{}

\newcommand{\sema}[1]{{\llbracket}#1{\rrbracket}}

\newtheorem{myclaim}[theorem]{Claim}

\newcommand{\proc}[1]{\ensuremath{\textbf{#1}}}

\newcommand{\aut}[1]{\mathcal{#1}}

\newcommand{\class}[1]{\ensuremath{\mathbb{#1}}}

\newcommand{\IB}{\class{IB}}
\newcommand{\IC}{\class{IC}}
\newcommand{\IP}{\class{IP}}
\newcommand{\IM}{\class{IM}}

\newcommand{\Inf}[1]{\ensuremath{\textit{Inf}_{#1}}}

\newcommand{\buchi}{B\"uchi}
\newcommand{\cobuchi}{coB\"uchi}
\newcommand{\initstate}{\ensuremath{q_\iota}}
\newcommand{\A}{{\aut{A}}}
\newcommand{\M}{{\aut{M}}}

\newcommand{\la}{\langle}
\newcommand{\ra}{\rangle}
\newcommand{\naturals}{\mathbb{N}}

\tikzset{
	setofstates/.style={
		rectangle,
		rounded corners,
		draw=black,
		text centered,
		inner sep=7
	},
}

\begin{document}

\title{Polynomial time algorithms for inclusion and equivalence of
  deterministic omega acceptors\thanks{This research was supported by grant 2016239 from the United States -- Israel Binational Science Foundation (BSF).}}
\titlerunning{Deterministic omega acceptor inclusion and equivalence in polynomial time}

\author{Dana Angluin\inst{1} \and Dana Fisman\inst{2}
  \institute{Yale University \and Ben-Gurion University}
}
 
  \maketitle              
  
  \begin{abstract}
  The class of omega languages recognized by deterministic parity
  acceptors (DPAs) or deterministic Muller acceptors (DMAs)
  is exactly the regular omega languages.  The inclusion
  problem is the following: given two acceptors
  $\aut{A}_1$ and $\aut{A}_2$, determine whether
  the language recognized by $\aut{A}_1$
  is a subset of the language recognized by $\aut{A}_2$,
  and if not, return an ultimately periodic omega word accepted
  by $\aut{A}_1$ but not $\aut{A}_2$.
  We describe polynomial time algorithms
  to solve this problem for two DPAs and for two DMAs.
  Corollaries include polynomial time algorithms
  to solve the equivalence problem for DPAs and DMAs,
  and also the inclusion and equivalence problems for
  deterministic \buchi\ and \cobuchi\ acceptors.
\end{abstract}

\section{Preliminaries}

We denote the set of nonnegative integers by $\naturals$,
and for any $m, n \in \naturals$, we let
$[m..n] = \{j \in \naturals \mid m \le j \wedge j \le n\}.$

\subsection{Words and Automata}

If $\Sigma$ is a finite alphabet of symbols, then
$\Sigma^*$ denotes the set of all finite words over $\Sigma$,
$\varepsilon$ denotes the empty word,
$|u|$ denotes the length of the finite word $u$,
and $\Sigma^+$ denotes the set of all nonempty finite words over $\Sigma$.
Also, $\Sigma^{\omega}$ denotes
the set of all infinite words over $\Sigma$, termed $\omega$-words.
For a finite or infinite word $w$, the successive symbols are
indexed by positive integers, and $w[i]$ denotes the symbol with
index $i$.
For words $u \in \Sigma^*$ and $v \in \Sigma^+$, $u(v)^{\omega}$ denotes
the ultimately periodic $\omega$-word consisting of $u$ followed by infinitely many
copies of $v$.

A \emph{complete deterministic automaton} is a tuple
$\M = \la \Sigma, Q, q_\iota ,\delta \ra$
consisting of a finite alphabet $\Sigma$ of symbols,
a finite set $Q$ of states,
an initial state $q_\iota\in Q$,
and a transition function ${\delta: Q \times \Sigma \rightarrow Q}$.
In this paper, \emph{automaton} will mean a complete deterministic automaton.
We extend $\delta$ to the domain $Q \times \Sigma^*$ inductively in the usual way,
and for $u \in \Sigma^*$, define $\M(u) = \delta(q_{\iota},u)$, the
state of $\M$ reached on input $u$.
Also, for a state $q \in Q$, $\M^q$ denotes the automaton
$\la \Sigma, Q, q ,\delta \ra$, in which the initial state
has been changed to $q$.

The {\it run} of an automaton $\M$
on an input $u \in \Sigma^*$ is the sequence
of states $q_0, q_1, \ldots, q_k$, where $k = |u|$, $q_0 = q_\iota$,
and for each $i \in [1..k]$,
$q_i = \delta(q_{i-1},u[i])$.
The {\it run} of $\M$ on an input $w \in \Sigma^{\omega}$ is the
infinite sequence of states $q_0, q_1, q_2, \ldots$, where $q_0 = q_\iota$, and
for each positive integer $i$,
$q_i = \delta(q_{i-1},w[i])$.
For an infinite word $w \in \Sigma^{\omega}$,
let $\Inf{\M}(w)$ denote the set of states of $\M$ that appear
infinitely often in the run of $\M$ on input $w$.

A state $q$ of an automaton $\M$ is \emph{reachable} if and only if
there exists a finite word $u \in \Sigma^*$ such that $\M(u) = q$.
We may restrict $\M$ to contain only its reachable states without
affecting its finite or infinite runs.

For any automaton
$\M = \la \Sigma, Q, q_\iota ,\delta \ra$
we may construct a related directed graph $G(\M) = (V,E)$ as follows.
The set $V$ of vertices is just the set of states $Q$, and there
is a directed edge $(q_1,q_2) \in E$ if and only if for some
symbol $\sigma \in \Sigma$ we have $\delta(q_1,\sigma) = q_2$.
By processing every pair $(q,\sigma) \in Q \times \Sigma$, the
set of edges $(q_1,q_2) \in E$ may be constructed in time
$O(|\Sigma| \cdot |Q|)$ using a hash table representation of $E$.

There are some differences in the terminology related to strong connectivity
between graph theory and omega automata, which we resolve as follows.
In graph theory,
a \emph{path} of length $k$ from $u$ to $v$
in a directed graph $(V,E)$ is a finite
sequence of vertices $v_0, v_1, \ldots, v_k$
such that
$u = v_0$, $v = v_k$ and
for each
$i$ with $i \in [1..k]$, $(v_{i-1},v_i) \in E$.
Thus, for every vertex $v$, there is a path of length $0$
from $v$ to $v$.
A set of vertices $S$ is \emph{strongly connected}
if and only if for all $u, v \in S$, there is a path of
some nonnegative length from $u$ to $v$
and all the vertices in the path are elements of $S$.
Thus, for every vertex $v$, the singleton set $\{v\}$ is
a strongly connected set of vertices.
A \emph{strongly connected component} of a directed graph is
a maximal strongly connected set of vertices.
There is a linear time algorithm to find the set of
strong components of a directed graph~\cite{Tarjan72}.

In the theory of omega automata,
a \emph{strongly connected component}
(SCC) of $A$ is a nonempty set $C \subseteq Q$ of states such that for any
$q_1, q_2 \in C$, there exists a nonempty word $u$ such that
$\delta(q_1,u) = q_2$,
and for every prefix $u'$ of $u$, $\delta(q_1,u') \in C$.
Note that a SCC of $A$ need not be maximal, and that a single state
$q$ of $A$ is not a SCC of $A$ unless for some symbol $\sigma \in \Sigma$
we have $\delta(q,\sigma) = q$.

In this paper, we use the terminology \emph{SCC} and \emph{maximal SCC}
to refer to the definitions from the theory of omega automata,
and the terminology 
\emph{graph theoretic strongly connected components} 
to refer to the definitions from graph theory.
Additionally, we use the term \emph{trivial strong component} to refer
to a graph theoretic strongly connected component that is a
singleton vertex $\{v\}$ such that there is no edge $(v,v)$.
Then if $\M$ is an automaton, the maximal SCCs of $\M$
are the graph theoretic strongly connected components of $G(\M)$
with the exception of the trivial strong components.
The following is a direct consequence of the definitions.

\begin{myclaim}
  For any automaton $\M$ and any $w \in \Sigma^{\omega}$,
  $\Inf{\M}(w)$ is a SCC of $\M$.
\end{myclaim}

\paragraph{The Product of Two Automata.}
Suppose ${\M}_1$ and ${\M}_2$ are automata with the same alphabet $\Sigma$,
where for $i = 1,2$,
${\M}_i = \la \Sigma, Q_i,(q_\iota)_i ,\delta_i \ra$.
Their product automaton, denoted ${\M}_1 \times {\M}_2$,
is the deterministic automaton
$\M = \la \Sigma, Q, q_\iota ,\delta \ra$
such that
$Q = Q_1 \times Q_2$, 
the set of ordered pairs of states of ${\M}_1$ and ${\M}_2$,
$q_\iota = ((q_\iota)_1,(q_\iota)_2)$, 
the pair of initial states of the two automata,
and for all $(q_1,q_2) \in Q$ and $\sigma \in \Sigma$,
$\delta((q_1,q_2), \sigma) = (\delta_1(q_1,\sigma),\delta_2(q_2, \sigma))$.
For $i = 1,2$, let $\pi_i$ be projection onto the $i$th coordinate,
so that for a subset $S$ of $Q$, 
$\pi_1(S) = \{q_1 \mid \exists q_2 (q_1,q_2) \in S\}$,
and analogously for $\pi_2$.

\subsection{Acceptors}

If $\M = \la \Sigma, Q, q_\iota,$ $\delta \ra$ is an automaton,
we may augment it with an acceptance condition $\alpha$ to
get a \emph{complete deterministic acceptor}
$\A = \la \Sigma, Q, q_\iota, \delta, \alpha \ra$,
which is
a machine that accepts some words and rejects others.
In this paper, an \emph{acceptor} will mean a complete deterministic acceptor.
If $q \in Q$, we use the notation $\A^q$ for the acceptor
$\la \Sigma, Q, q, \delta, \alpha \ra$, in which the
initial state has been changed to the state $q$.
An acceptor accepts a word if the run on that word is accepting,
as defined below for the types of acceptors we consider.
For finite words the acceptance condition is a set $F \subseteq Q$
and the run on a word $v \in \Sigma^*$ is accepting
iff it ends in an accepting state,
that is, $\M(v) \in F$.
We use DFA to denote the class of acceptors of
finite words, and \class{DFA} for the languages they accept, which
is the class of regular languages.

For $\omega$-words $w$, 
there are various acceptance conditions in the literature;
we consider four of them: \buchi, \cobuchi, parity, and Muller, which
are all based on $\Inf{\M}(w)$, the set of states visited infinitely
often in the run of the automaton on the input $w$.

The \buchi\ and \cobuchi\ acceptance conditions are also specified by
a set $F \subseteq Q$.
The run of a \buchi\ acceptor on an input word $w \in \Sigma^{\omega}$
is accepting iff it visits at least one state in $F$ infinitely often,
that is, $\Inf{\M}(w) \cap F \neq \emptyset$.
The run of a \cobuchi\ acceptor on an input word $w \in \Sigma^{\omega}$
is accepting iff it visits $F$ only finitely many times,
that is, $\Inf{\M}(w) \cap F = \emptyset$.

A parity acceptance condition is a map
$\kappa:Q \rightarrow \naturals$ assigning to each state a natural number
termed a color (or priority).
We extend $\kappa$ to sets of states in the natural way,
that is, for $S \subseteq Q$, $\kappa(S) = \{\kappa(q) \mid q \in S\}$.
For a parity acceptor $\aut{P}$ and an $\omega$-word $w$,
we denote by $\aut{P}(w)$ the minimum color of all states
visited infinitely often by $\aut{P}$ on input $w$, that is,
\[\aut{P}(w) = \min(\kappa(\Inf{\M}(w))).\]
The run of a parity acceptor $\aut{P}$
on an input word $w \in \Sigma^{\omega}$
is accepting iff the minimum color visited infinitely often is odd,
that is, $\aut{P}(w)$ is odd.

A Muller\ acceptance condition is a family of final state sets
$\mathcal{F}=\{F_1,\ldots,F_k\}$ for some $k\in\naturals$ and
$F_i\subseteq Q$ for $i \in [1..k]$.
The run of a Muller acceptor is accepting
iff the set of states visited infinitely often in the run
is an element of $\mathcal{F}$,
that is, $\Inf{M}(w) \in \mathcal{F}$.

To measure the running times of algorithms taking acceptors
as inputs, we specify size measures for these acceptors.
For an automaton
$\M = \la \Sigma, Q, q_\iota, \delta \ra$,
we specify its size as $|Q| \cdot |\Sigma|$, which is the
number of entries in a table that explicitly specifies the
transition function $\delta$.
Note that the size of the product automaton $\M_1 \times \M_2$
is at most $|\Sigma| \cdot |Q_1| \cdot |Q_2|$, where for $i = 1,2$,
$Q_i$ is the set of states of $\M_i$.
For a deterministic \buchi, \cobuchi, or parity acceptor,
the space to specify the acceptance condition is dominated
by the size of the automaton, so that is the size of the acceptor.
However, for deterministic Muller acceptors, we add the quantity
$|Q| \cdot |\mathcal{F}|$, which bounds the size of a table
explicitly specifying the membership of each state in each
final state set in $\mathcal{F}$.\footnote{See~\cite{Boker17} for
a discussion regarding size of Muller automata in the literature.}

We use $\sema{\A}$ to denote the set of words 
accepted by a given acceptor $\A$.
Two acceptors $\aut{A}_1$ and $\aut{A}_2$ are \emph{equivalent}
iff $\sema{\aut{A}_1}=\sema{\aut{A}_2}$.
We use DBA, DCA, DPA, and DMA for the classes of (deterministic)
\buchi, \cobuchi, parity, and Muller acceptors.
We use \class{DBA}, \class{DCA}, \class{DPA}, and \class{DMA}
for the classes of languages they recognize, respectively.
\class{DPA} and \class{DMA} are each 
the full class of regular $\omega$-languages,
while \class{DBA} and \class{DCA} are proper subclasses.

We observe the following known facts 
about the relationships between DBAs, DCAs and DPAs.

\begin{myclaim}
  \label{DBA-DCA-complement}
  Let $\aut{B} = \aut{C} = \la \Sigma, Q, q_\iota, \delta, F \ra$, where
  $\aut{B}$ is a DBA and $\aut{C}$ is a DCA.
  Then the languages recognized by $\aut{B}$ and $\aut{C}$ are complements
  of each other, that is,
  $\sema{\aut{B}} = \Sigma^{\omega} \setminus \sema{\aut{C}}$.
\end{myclaim}

\begin{myclaim}
  \label{DBA-to-DPA}
    Let the DBA $\aut{B} = \la \Sigma, Q, q_\iota, \delta, F \ra$.
    Define the coloring $\kappa_{\aut{B}}(q) = 1$ for all $q \in F$ and
    $\kappa_{\aut{B}}(q) = 2$ for all $q \in (Q \setminus F)$.
    Define the DPA
    $\aut{P} =  \la \Sigma, Q, q_\iota, \delta, \kappa_{\aut{B}} \ra$.
    Then $\aut{B}$ and $\aut{P}$ accept the same language, that is,
    $\sema{\aut{B}} = \sema{\aut{P}}$.
\end{myclaim}

Analogously, for DCAs we have the following.

\begin{myclaim}
  \label{DCA-to-DPA}
    Let the DCA $\aut{B} = \la \Sigma, Q, q_\iota, \delta, F \ra$.
    Define the coloring $\kappa_{\aut{C}}(q) = 0$ for all $q \in F$ and
    $\kappa_{\aut{C}}(q) = 1$ for all $q \in (Q \setminus F)$.
    Define the DPA
    $\aut{P} =  \la \Sigma, Q, q_\iota, \delta, \kappa_{\aut{C}} \ra$.
    Then $\aut{C}$ and $\aut{P}$ accept the same language, that is,
    $\sema{\aut{C}} = \sema{\aut{P}}$.
\end{myclaim}
	
\subsection{Right congruences}

An equivalence relation $\sim$ on $\Sigma^*$ is a \emph{right congruence}
if $x\sim y$ implies $xv \sim yv$ for every $x,y,v\in\Sigma^*$.
The \emph{index} of $\sim$, denoted ${|\!\sim\!|}$
is the number of equivalence classes of $\sim$.

Given an automaton
$\M=\la \Sigma, Q, \initstate, \delta \ra$,
we can associate with it a right congruence as follows:
$x \sim_\M y$ iff $M$ reaches the same state when reading $x$ or $y$,
that is, $\M(x) = \M(y)$.
If all the states of $\M$ are reachable then
the index of $\sim_\M$ is exactly the number of states of $\M$.

Given a language $L\subseteq \Sigma^*$,
its \emph{canonical right congruence} $\sim_L$ is defined as follows: 
$x \sim_L y$ iff ${\forall z \in \Sigma^*}$ 
we have ${xz\in L} \iff {yz \in L}$. 
For a word $v\in\Sigma^*$,
the notation $[v]$ is used for the equivalence class
of $\sim$ in which $v$ resides.
	
With a right congruence $\sim$ of finite index
one can naturally associate an automaton
$\M_\sim=\la \Sigma, Q, \initstate, \delta \ra$ as follows.
The set of states $Q$ consists of the equivalence classes of $\sim$.
The initial state $\initstate$ is the equivalence class $[\varepsilon]$.
The transition function $\delta$ is defined by $\delta([u],a)=[ua]$.

The Myhill-Nerode Theorem states
that a language $L \subseteq \Sigma^*$ is regular
iff $\sim_L$ is of finite index.
Moreover, if $L$ is accepted by a DFA $\A$ with automaton $\M$,
then $\sim_\M$ refines $\sim_L$.
Finally, the index of $\sim_L$ gives
the number of states of the minimal DFA for $L$. 
	
For an $\omega$-language $L\subseteq \Sigma^\omega$,
the right congruence $\sim_L$ is defined analogously,
by quantifying over $\omega$-words.
That is, $x \sim_L y$ iff ${\forall z \in \Sigma^\omega}$
we have ${xz\in L} \iff {yz \in L}$.

For a regular $\omega$-language $L$, the right congruence relation
$\sim_L$ is always of finite index, and we may define the
\emph{right congruence automaton} for $L$ to be the automaton $M_{\sim_L}$.
However, in constrast to the Myhill-Nerode Theorem,
the right congruence automaton for $L$ may not be adequate
to support an acceptor for $L$.
As an example consider the language $L = (a+b)^*(bba)^\omega$.
We have that $\sim_{L}$ consists of just one equivalence class, 
since for any $x\in\Sigma^*$ and $w\in\Sigma^\omega$ we have
that $xw \in L$ iff $w$ has $(bba)^\omega$ as a suffix.
However, a DPA or DMA recognizing $L$ needs more than a single state.
	
The classes $\IB$, $\IC$, $\IP$, $\IM$ of omega languages are
defined as those for which the right congruence automaton
is adequate to support an acceptor of the corresponding type.
A language $L$ is in $\IB$ (resp., $\IC$, $\IP$, $\IM$)
if there exists a DBA (resp., DCA, DPA, DMA) $\aut{A}$ such that 
$L=\sema{\aut{A}}$ and the automaton part of $\aut{A}$ is isomorphic
to the right congruence automaton of $L$.
These classes are more expressive than one might conjecture;
it was shown in~\cite{AngluinF18} that
in every class of the infinite Wagner hierarchy~\cite{Wagner75}
there are languages in $\IP$. 

\subsection{The inclusion and equivalence problems}

The \emph{inclusion problem} for two $\omega$-acceptors is the following.
Given as input two $\omega$-acceptors  $\aut{A}_1$ and $\aut{A}_2$
over the same alphabet,
determine whether the language accepted by $\aut{A}_1$ is
a subset of the language accepted by $\aut{A}_2$, that is,
whether $\sema{\aut{A}_1} \subseteq \sema{\aut{A}_2}$.
If so, the answer should be ``yes''; if not, the answer should
be ``no'' and a \emph{witness}, that is, an ultimately periodic
$\omega$-word $u(v)^{\omega}$ accepted by $\aut{A}_1$ but
rejected by $\aut{A}_2$.

The \emph{equivalence problem} for two $\omega$-acceptors is similar: 
the input is two $\omega$-acceptors $\aut{A}_1$ and $\aut{A}_2$ 
over the same alphabet, and the problem is to determine 
whether they are equivalent,
that is, whether $\sema{\aut{A}_1} = \sema{\aut{A}_2}$.
If so, the answer should be ``yes''; if not, the answer should
be ``no'' and a witness, that is, an ultimately periodic
$\omega$-word $u(v)^{\omega}$ that is accepted by one of
the two acceptors and not accepted by the other.

Clearly, if we have a procedure to solve the inclusion problem,
at most two calls to it will solve the equivalence problem.
Thus, we focus on the inclusion problem.
By Claim~\ref{DBA-DCA-complement}, the inclusion and
equivalence problems for DCAs are efficiently reducible
to those for DBAs, and vice versa.
Also, by Claims~\ref{DBA-to-DPA} and \ref{DCA-to-DPA}, the inclusion
and equivalence problems for DBAs and DCAs are efficiently
reducible to those for DPAs.
Hence we consider the problems of inclusion for two DPAs and
inclusion for two DMAs.

Note that 
while a polynomial algorithm for testing inclusion of DFAs can be obtained
using polynomial algorithms for complementation, intersection and emptiness
(since for any two languages $L_1 \subseteq L_2$
if and only if $L_1 \cap \overline{L_2}=\emptyset$),
a similar approach does not work in the case of DPAs;
although complementation and emptiness for DPAs
can be computed in polynomial time, 
intersection cannot~\cite[Theorem 9]{Boker18}.

For the inclusion problem for DBAs, DCAs, and DPAs, 
Schewe~\cite{Schewe10} gives the following result.
\begin{theorem}
  \label{theorem:schewe}
  The inclusion problems for DBAs, DCAs, and DPAs are in NL.
\end{theorem}
Because NL (nondeterministic logarithmic space) is contained 
in polynomial time,
this implies the existence of polynomial time inclusion and equivalence
algorithms for DBAs, DCAs, and DPAs.
Because Schewe does not give a proof or reference 
for Theorem~\ref{theorem:schewe},
for completeness we include a brief proof.
\begin{proof}
  For $i = 1,2$, let
  $\aut{P}_i = \la \Sigma, Q_i, (q_\iota)_i, \delta_i, \kappa_i \ra$ be a DPA.
  It suffices to guess two states $q_1 \in Q_1$ and $q_2 \in Q_2$,
  and two words $u \in \Sigma^*$ and $v \in \Sigma^+$,
  and to check that for $i = 1,2$, $\delta_i((q_\iota)_i,u) = q_i$
  and $\delta_i(q_i,v) = q_i$, and also,
  that the smallest value of $\kappa_1(q)$ in the loop 
  in $\aut{P}_1$ from $q_1$ to $q_1$ on input
  $v$ is odd,
  while the smallest value of $\kappa_2(q)$ in the loop 
  in $\aut{P}_2$ from $q_2$ to $q_2$ on input 
  $v$ is even.  
  Logarithmic space is enough to record the two
  guessed states $q_1$ and $q_2$ as well as 
  the current minimum values of $\kappa_1$ and $\kappa_2$ 
  as the loops on $v$ are traversed in the two automata.
  The words $u$ and $v$ need only be guessed symbol-by-symbol, 
  using a pointer in each automaton to keep track of its current state.
\end{proof}

This approach does not seem to work in the case of testing DMA inclusion,
because the acceptance criterion for DMAs would require keeping track of
the set of states traversed in the loop on the word $v$, which would
in general require more than logarithmic space.
In this paper, we give an explicit polynomial time algorithm for testing
DPA inclusion, as well as a polynomial time algorithm for testing
DMA inclusion, a novel result.

\section{An inclusion algorithm for DPAs}

In this section we describe an explicit polynomial time algorithm 
for the inclusion problem for two DPAs.

\subsection{Constructing a witness}

In order to return a witness,
the inclusion algorithm calls a procedure \proc{Witness}
that takes as input an automaton $\M = \la \Sigma, Q, q_\iota, \delta \ra$
with no unreachable states and 
a SCC $C$ of $\M$ and returns
an ultimately periodic word $u(v)^\omega$ such that
$\Inf{\M}(u(v)^\omega) = C$.
The procedure first chooses a state $q \in C$ and uses breadth
first search in $G(\M)$ to find a shortest finite word $u$ such
that $\M(u) = q$.
The length of $u$ is at most $|Q|-1$.

If $|C| = 1$, then the procedure finds a symbol $\sigma \in \Sigma$
such that $\delta(q,\sigma) = q$, and returns $u(\sigma)^\omega$.
Otherwise, for every $q' \in C$ such that $q' \neq q$, it uses
breadth first search in the subgraph of $G(\M)$ induced by
the vertices in $C$ to find shortest finite words $v_{q,q'}$ and
$v_{q',q}$ such that $\delta(q,v_{q,q'}) = q'$ and $\delta(q',v_{q',q}) = q$,
and no intermediate state is outside of $C$.
The finite word $v$ is the concatenation, 
over all $q' \in C$ with $q' \neq q$
of the finite words $v_{q,q'} \cdot v_{q',q}$.
The length of $v$ is at most $O(|C|^2)$.

In $\M$, the input $u$ reaches $q$,
and from $q$ the word $v$ visits no state outside
of $C$, visits each state of $C$, and returns to $q$.
Thus, $\Inf{\M}(u(v)^\omega) = C$, as desired.
The length of $u(v)^\omega$ is bounded by $O(|Q|+|C|^2)$.
We have proved the following.

\begin{lemma}
  \label{lem:C-to-witness}
  The \proc{Witness} procedure takes as input
  an automaton $\M$ with no unreachable states
  and a SCC $C$ of $\M$, and returns 
  an ultimately periodic word $u(v)^\omega$
  such that $\Inf{\M}(u(v)^\omega) = C$.

  The length of $u(v)^\omega$ is bounded by $O(|Q|+|C|^2)$, 
  where $Q$ is the set of states of $\M$.
  The running time of the \proc{Witness} procedure is
  bounded by this quantity plus $O(|Q|\cdot|\Sigma|)$.
\end{lemma}

\subsection{Searching for $w$ with $\aut{P}_1(w) = k_1$ 
and $\aut{P}_2(w) = k_2$}

We next describe a procedure \proc{Colors} that takes as input
two DPAs over the same alphabet,
${\aut{P}}_i = \la \Sigma, Q_i,(q_\iota)_i ,\delta_i, \kappa_i \ra$
for $i = 1,2$,
and two nonnegative integers $k_1$ and $k_2$,
and answers the question of whether there exists an $\omega$-word $w$
such that for $i = 1,2$,
$\aut{P}_i(w) = k_i$,
that is, the minimum color of the states
visited infinitely often on input $w$
in $\aut{P}_1$ is $k_1$ and in $\aut{P}_2$ is $k_2$.
If there is no such $w$, the return value will be ``no'', but
if there is such a $w$, the return value will be
``yes'' and an ultimately periodic witness $u(v)^{\omega}$
such that $\aut{P}_1(u(v)^{\omega}) = k_1$
and $\aut{P}_2(u(v)^{\omega}) = k_2$.

For $i =1,2$, let $\M_i$ be the automaton of $\aut{P}_i$,
that is,
$\M_i = \la \Sigma, Q_i,(q_\iota)_i ,\delta_i \ra$.
The \proc{Colors} procedure constructs
the product automaton $\M = {\M}_1 \times {\M}_2$,
eliminating unreachable states.
It then constructs the related directed graph $G(\M)$,
and the subgraph $G'$ of $G(\M)$ obtained 
by removing all vertices $(q_1,q_2)$ 
such that $\kappa_1(q_1) < k_1$ or $\kappa_2(q_2) < k_2$
and their incident edges.

In linear time in the size of $G'$, 
the \proc{Colors} procedure computes the graph theoretic
strongly connected components of $G'$ and eliminates any
trivial strong components.
The procedure then loops through 
the nontrivial strong components $C$ of $G'$ 
checking whether
$\min(\kappa_1(\pi_1(C))) = k_1$ and
$\min(\kappa_2(\pi_2(C))) = k_2$.
If so, the \proc{Witness} procedure is called with
inputs $\M$ and $C$, and the resulting witness $u(v)^{\omega}$
is returned with the answer ``yes''.
If none of the nontrivial strong components of $G'$ satisfies
this condition, then the answer ``no'' is returned.

\begin{theorem}
  \label{thm:dpa-k1-k2}
  The \proc{Colors} procedure takes as input two arbitrary
  DPAs $\aut{P}_1$ and $\aut{P}_2$ over the same alphabet 
  $\Sigma$
  and
  two nonnegative integers $k_1$ and $k_2$, and
  determines whether there exists an $\omega$-word $w$ such that
  for $i = 1,2$, $\aut{P}_i(w) = k_i$, returning the
  answer ``no'' if not, and returning the answer ``yes'' and a
  witness word $w = u(v)^{\omega}$ such that $P_i(u(v)^\omega) = k_i$
  for $i = 1,2$, if so.

  The length of a witness $u(v)^\omega$ is
  bounded by $O((|Q_1|\cdot|Q_2|)^2)$ 
  and the running time of the \proc{Colors} procedure
  is bounded by this quantity plus $O(|\Sigma|\cdot|Q_1|\cdot|Q_2|)$,
  where $Q_i$ is the set of states of $\aut{P}_i$ for $i = 1,2$.
\end{theorem}

\begin{proof}
Assume first that the \proc{Colors} procedure returns ``yes'' 
with a witness $u(v)^\omega$.
This occurs only if the procedure finds a nontrivial graph theoretic strong
component $C$ of $G'$ such that 
$\min(\kappa_1(\pi_1(C))) = k_1$ and
$\min(\kappa_2(\pi_2(C))) = k_2$.
For $i = 1,2$,
$\pi_i(C)$ is the set of states visited infinitely often by ${\M}_i$ on
the input $u(v)^{\omega}$, which has minimum color $k_i$.
That is, for $i = 1,2$, $\aut{P}_i(u(v)^{\omega}) = k_i$,
and $u(v)^{\omega}$ is a correct witness for the answer ``yes''.

To see that the \proc{Colors} procedure 
does not incorrectly answer ``no'', we argue as follows.
Suppose $w$ is an $\omega$-word such that for $i = 1,2$,
$\aut{P}_i(w) = k_i$, that is,
if $C_i = \Inf{{\M}_i}(w)$ then $\min(\kappa_i(C_i)) = k_i$.
Clearly, no state in $C_i$ has a color less than $k_i$, so if
$C = \Inf{\M}(w)$ is the set of states visited infinitely 
often in $\M$ on input $w$, all the elements of $C$ will be in $G'$.
Because $C$ is a SCC of $\M$,
$C$ is a nontrivial graph theoretic strongly connected 
set of vertices of $G'$.
Thus $C$ is contained in 
a graph theoretic nontrivial (maximal) strong component $C'$ of $G'$,
and because there are no vertices $(q_1,q_2)$ in $G'$
with $\kappa_1(q_1) < k_1$ or $\kappa_2(q_2) < k_2$, we must have
$\min(\kappa_i(\pi_i(C'))) = k_i$ for $i = 1,2$.
Thus, the algorithm will find at least one such graph theoretic strong
component $C'$ of $G'$ and return ``yes'' and a correct witness.

The running time of the \proc{Colors} procedure, exclusive of a
call to the \proc{Witness} procedure, is linear in the
size of $\M$, that is, $O(|\Sigma| \cdot |Q_1| \cdot |Q_2|)$.
Because the number of states of $\M$ is bounded by $|Q_1|\cdot|Q_2|$,
this is also an upper bound on the size of any nontrivial
graph theoretic strong component of $G(\M)$, so the the
length of any witness $u(v)^\omega$ returned is bounded
by $O((|Q_1|\cdot|Q_2|)^2)$ and the overall running time is
bounded by this quantity plus $O(|\Sigma|\cdot|Q_1|\cdot|Q_2|)$.
\end{proof}

\subsection{Inclusion and equivalence algorithms for DPAs}

The inclusion problem for DPAs $\aut{P}_1$ and $\aut{P}_2$ 
over the same alphabet can be solved by
looping over all odd $k_1$ in the range of $\kappa_1$ and
all even $k_2$ in the range of $\kappa_2$,
calling the \proc{Colors} procedure with inputs
$\aut{P}_1$, $\aut{P}_2$, $k_1$, and $k_2$.
If the \proc{Colors} procedure returns 
any witness $u(v)^\omega$, then
$u(v)^{\omega} \in \sema{\aut{P}_1} \setminus \sema{\aut{P}_2}$,
and $u(v)^\omega$ is returned as a witness of non-inclusion.
Otherwise, by Theorem~\ref{thm:dpa-k1-k2}, there is no $\omega$-word $w$
accepted by $\aut{P}_1$ and not accepted by $\aut{P}_2$, and
the answer ``yes'' is returned for the inclusion problem.
Note that for $i = 1,2$,
the range of $\kappa_i$ has at most $|Q_i|$ distinct elements.

\begin{corollary}
  \label{cor:DPA-inclusion-equivalence}
  There are algorithms for the inclusion and equivalence problems
  for two DPAs $\aut{P}_1$ and $\aut{P}_2$ over the same alphabet $\Sigma$
  that run in time bounded by 
  $O((|Q_1|\cdot|Q_2|)^3 + |\Sigma|(|Q_1|\cdot|Q_2|)^2)$,
  where $Q_i$ is the set of states of $\aut{P}_i$ for $i = 1,2$.
  A returned witness $u(v)^\omega$ has length $O((|Q_1|\cdot|Q_2|)^2)$.
\end{corollary}

From Claims~\ref{DBA-to-DPA} and \ref{DCA-to-DPA}, and the fact that
two colors suffice in the transformation of a DBA (or DCA) to a DPA, 
we have the following.

\begin{corollary}
  \label{cor:DBA-DCA-inclusion-equivalence}
  There are algorithms for the inclusion and equivalence problems
  for two DBAs (or DCAs) $\aut{B}_1$ and $\aut{B}_2$ over the same
  alphabet $\Sigma$ that run in time bounded by
  $O((|Q_1|\cdot|Q_2|)^2 + |\Sigma|\cdot|Q_1|\cdot|Q_2|)$, where
  $Q_i$ is the set of states of $\aut{B}_i$ for $i = 1,2$.
  A returned witness $u(v)^\omega$ has length $O((|Q_1|\cdot|Q_2|)^2)$.
\end{corollary}

\section{Inclusion and equivalence algorithms for DMAs}

In this section we develop a polynomial time algorithm
to solve the inclusion problem for two DMAs over the same alphabet.
The proof proceeds in two parts: (1) a polynomial time reduction
of the inclusion problem for two DMAs to the inclusion problem
for a DBA and a DMA, and (2) a polynomial time algorithm for
the inclusion problem for a DBA and a DMA.

\subsection{Reduction of DMA inclusion to DBA/DMA inclusion}

We first reduce the problem of inclusion for two arbitrary DMAs
to the inclusion problem for two DMAs 
where the first one has just a single final state set.
For $i = 1,2$, define the DMA
$\aut{U}_i = \la Q_i, \Sigma, (q_\iota)_i, \delta_i, \mathcal{F}_i \ra$,
where $\aut{F}_i$ is the set of final state sets for $\aut{U}_i$.
Let the elements of $\mathcal{F}_1$ be $\{F_1, \ldots, F_k\}$,
and for each $j \in [1..k]$, let
\[\aut{U}_{1,j} = \la Q_1, \Sigma, (q_\iota)_1, \delta_1, \{F_j\} \ra,\]
that is, 
$\aut{U}_{1,j}$ is $\aut{U}_1$ with $F_j$ as its only final state set.
Then by the definition of DMA acceptance,
\[\sema{\aut{U}_1} = \bigcup_{j=1}^k \sema{\aut{U}_{1,j}},\]
which implies that to test whether
$\sema{\aut{U}_1} \subseteq \sema{\aut{U}_2}$,
it suffices to test for all $j \in [1..k]$ that
$\sema{\aut{U}_{1,j}} \subseteq \sema{\aut{U}_2}$.

\begin{myclaim}
  \label{claim:arbitrary-to-one-final-state-set}
  Suppose $A$ is a procedure that solves the
  inclusion problem for two DMAs over the same alphabet,
  assuming that the first DMA has a single final state set.
  Then there is an algorithm that solves the inclusion
  problem for two arbitrary DMAs over the same alphabet,
  say $\aut{U}_1$ and $\aut{U}_2$,
  which simply makes $|\mathcal{F}_1|$ calls to $A$, where $\mathcal{F}_1$
  is the family of final state sets of $\aut{U}_1$.
\end{myclaim}

Next we describe a procedure \proc{SCC-to-DBA} that
takes as inputs an automaton $\M$, a SCC $F$ of $\M$,
and a state $q \in F$, and returns
a DBA $B(\M,F,q)$ that accepts exactly $L(\M,F,q)$,
where $L(\M,F,q)$ is the set of
$\omega$-words $w$ that visit only the
states of $F$ when processed by $\M$ starting at state $q$,
and visits each of them infinitely many times.

Assume the states in $F$ are $\{q_0, q_1, \ldots, q_{m-1}\}$,
where $q_0 = q$.
The DBA $B(\M, F, q)$ is $\la Q', \Sigma, q_0, \delta', \{q_0\} \ra$,
where we define $Q'$ and $\delta'$ as follows.
We create new states $r_{i,j}$ for $i,j \in [0..m-1]$ such that
$i \neq j$, and denote the set of these by $R$.
We also create a new dead state $d_0$.
Then the set of states $Q'$ is $Q \cup R \cup \{d_0\}$.

For $\delta'$,
the dead state $d_0$ behaves as expected:
for all $\sigma \in \Sigma$, $\delta'(d_0,\sigma) = d_0$.
For the other states in $Q'$, let
$\sigma \in \Sigma$ and $i \in [0..m-1]$.
If $\delta(q_i, \sigma)$ is not in $F$, then
in order to deal with runs that would visit
states outside of $F$, we define
$\delta'(q_i,\sigma) = d_0$ and, for all $j \neq i$,
$\delta'(r_{i,j},\sigma) = d_0$.

Otherwise, for some $k \in [0..m-1]$ we have
$q_k = \delta(q_i,\sigma)$.
If $k = (i+1) \bmod m$, then we define
$\delta'(q_i,\sigma) = q_k$, and otherwise we
define
$\delta'(q_i,\sigma) = r_{k,(i+1) \bmod m}$.
For all $j \in [0..m-1]$ with $j \neq i$,
if $k = j$, we define $\delta'(r_{i,j},\sigma) = q_k$,
and otherwise we define
$\delta'(r_{i,j},\sigma) = r_{k,j}$.

Intuitively, for an input from $L(\M,F,q)$, in $B(\M,F,q)$
the states $q_i$ are visited in a repeating cyclic
order: $q_0, q_1, \ldots, q_{m-1}$, and the meaning of the
state $r_{i,j}$ is that at this point in the input,
$\M$ would be in state $q_i$, and the machine $B(\M,F,q)$
is waiting for a transition that would arrive at state $q_j$ in $\M$,
in order to proceed to state $q_j$ in $B(\M,F,q)$.\footnote{This construction is reminiscent of the construction transforming a generalized B\"uchi into a B\"uchi automaton~\cite{Vardi08,Choueka74}, by considering each state in $F$ as a singelton set of a generalized B\"uchi, but here we need to send transitions to states outside $F$ to a sink state.}
An example of the construction is shown in Fig.~\ref{fig:example-for-B-M-F-q};
the dead state and unreachable states are omitted for clarity.

\begin{figure}[t]
  \begin{center}
    \scalebox{1.0}{
      \begin{tikzpicture}[->,>=stealth',shorten >=1pt,auto,node
	  distance=1.8cm,semithick,initial text=,initial where=left]

        \node (start) {$\M:$};
	\node[state] (q0) [right of=start] {$q_0$};
	\node[state] (q1) [right of=q0] {$q_1$};
	\node[state] (q2) [below of=q1] {$q_2$};
	\node[state] (q3) [below of=q0] {$q_3$};

	\path (q0) edge [loop above]        node {$b$} (q0);
	\path (q0) edge [above, near start] node {$a$} (q2);
	\path (q1) edge [above]             node {$a$} (q0);
	\path (q1) edge [bend right]        node {$b$} (q2);
	\path (q2) edge [bend right]        node {$a$} (q1);
	\path (q2) edge                     node {$b$} (q3);
	\path (q3) edge [below, near start] node {$a$} (q1);
	\path (q3) edge                     node {$b$} (q0);

        \node (belowstart) [below of=start] {};
        \node (belowbelowstart) [below of=belowstart] {};
        \node (next) [below of=belowbelowstart] {$B(\M,F,q):$}; 

        \node[state,initial,accepting] (bq0) [right of=next] {$q_0$};
	\node[state] (br01) [above right of=bq0] {$r_{0,1}$};
	\node[state] (br21) [right of=bq0] {$r_{2,1}$};
	\node[state] (bq1)  [right of=br21] {$q_1$};
        \node[state] (br02) [below right of=bq1] {$r_{0,2}$};
        \node[state] (bq2) [below left of=br02] {$q_2$};
        \node[state] (br10) [left of=bq2] {$r_{1,0}$};
        \node[state] (br20) [left of=br10] {$r_{2,0}$};

	\path (bq0) edge [above]        node {$b$} (br01);
	\path (bq0) edge [below]        node {$a$} (br21);
        \path (br01) edge [loop above]  node {$b$} (br01);
        \path (br01) edge [right]       node {$a$} (br21);
        \path (br21) edge [above]       node {$a$} (bq1);
        \path (bq1) edge [right]        node {$a$} (br02);
        \path (bq1) edge [right]        node {$b$} (bq2);
        \path (br02) edge [loop right]  node {$b$} (br02);
        \path (br02) edge [below]       node {$a$} (bq2);
        \path (bq2) edge [above]        node {$a$} (br10);
        \path (br10) edge [right]       node {$a$} (bq0);
        \path (br10) edge [below, bend left] node {$b$} (br20);
        \path (br20) edge [above, bend left] node {$a$} (br10);
        
    \end{tikzpicture} }   
    \caption{Example of the construction of $B(\M,F,q)$ with $F = \{q_0,q_1,q_2\}$ and $q = q_0$.}
    \label{fig:example-for-B-M-F-q}
  \end{center}
\end{figure}
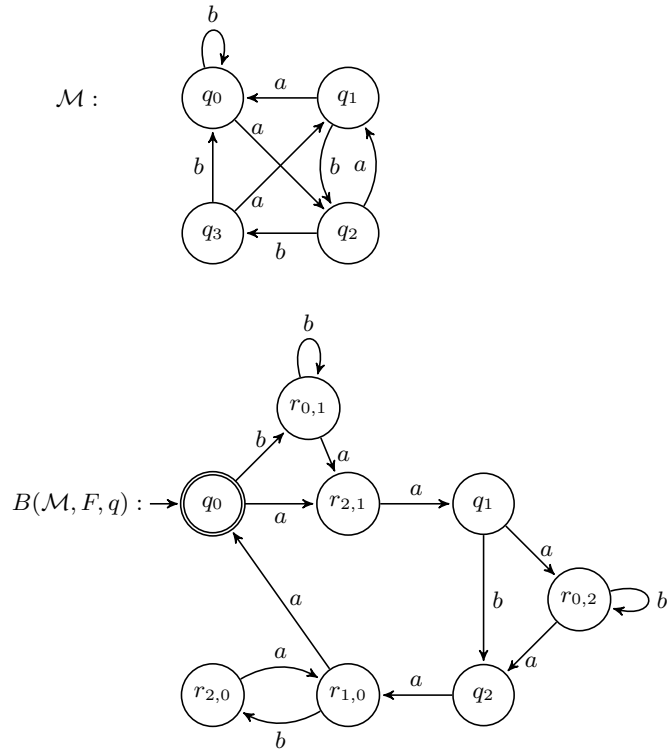

\begin{lemma}
  Let $\M$ be an automaton with alphabet $\Sigma$ and
  states $Q$, and let $F$ a SCC of $\M$ and $q \in F$.
  With these inputs, the procedure \proc{SCC-to-DBA}
  returns the DBA $B(\M,F,q)$, which accepts the language $L(\M,F,q)$
  and has $|F|^2+1$ states.
  The running time of \proc{SCC-to-DBA} is
  $O(|\Sigma|(|Q|+|F|^2))$.
\end{lemma}

\begin{proof}
  Suppose $w$ is in $L(\M,F,q)$.
  Let $q = s_0, s_1, s_2, \ldots$ be the sequence of states in the
  run of $\M$ from state $q$ on input $w$.
  This run visits only states in $F$ and visits each one of them
  infinitely many times.
  We next define a particular increasing sequence 
  $i_{k,\ell}$ of indices in $s$,
  where $k$ is a positive integer and $\ell \in [0,m-1]$.
  These indices mark particular visits to the states 
  $q_0, q_1, \ldots, q_{m-1}$
  in repeating cyclic order.
  The initial value is $i_{1,0} = 0$, marking the initial visit to $q_0$.
  If $i_{k,\ell}$ has been defined and $\ell < m-1$, then $i_{k,\ell+1}$
  is defined as the least natural number $j$ such that $j > i_{k,\ell}$
  and $s_j = q_{\ell+1}$, marking the next visit to $q_{\ell+1}$.
  If $\ell = m-1$, then $i_{k+1,0}$ is defined
  as the least natural number $j$ such that $j > i_{k,\ell}$ and
  $s_j = q_0$, marking the next visit to $q_0$.

  There is a corresponding division of $w$ into a concatenation
  of finite segments $w_{1,1}, w_{1,2}, \ldots, w_{1,m-1}, w_{2,0}, \ldots$
  between consecutive elements in the increasing sequence of indices.
  An inductive argument shows that in $B(\M,F,q)$, the
  prefix of $w$ up through $w_{k,\ell}$ arrives at the state $q_\ell$,
  so that $w$ visits $q_0$ infinitely often and is therefore
  accepted by $B(\M,F,q)$.

  Conversely, suppose $B(\M,F,q)$ accepts the $\omega$-word $w$.
  Let $s_0, s_1, s_2, \ldots$ be the run of $B(\M,F,q)$ on $w$,
  and let $t_0, t_1, t_2, \ldots$ be the run of $\M$ starting
  from $q$ on input $w$.
  An inductive argument shows that if $s_n = q_i$ then $t_n = q_i$, and
  if $s_n = r_{i,j}$ then $t_n = q_i$.
  Because the only way the run $s_0, s_1, \ldots$ can visit the final
  state $q_0$
  infinitely often is to progress through the states
  $q_0, q_1, \ldots q_{m-1}$ in repeating cyclic order, the run
  $t_0, t_1, \ldots$ must visit only states in $F$ and visit
  each of them infinitely often, so $w \in L(\M,F,q)$.
  
  The DBA $B(\M,F,q)$ has a dead state, and $|F|$ states for each
  element of $F$, for a total of $|F|^2 + 1$ states.
  The running time of \proc{SCC-to-DBA} is linear in the size of $\M$ and
  the size of the resulting DBA, that is,
  $O(|\Sigma|(|Q| + |F|^2))$, polynomial in the size of $\M$.
\end{proof}

We now show that this construction may be used to reduce the inclusion of
two DMAs to the inclusion of a DBA and a DMA.
Recall that if $\A$ is an acceptor and $q$ is a state of $\A$,
then $\A^q$ denotes the acceptor $\A$ with the initial state changed
to $q$.

\begin{lemma}
  Let $\aut{U}_1$ be a DMA with automaton $\M_1$ and a single final
  state set $F_1$.
  Let $\aut{U}_2$ be an arbitrary DMA over the same alphabet as $\aut{U}_1$,
  with automaton $\M_2$ and family of final state sets $\mathcal{F}_2$.
  Let $\M$ denote the product automaton $\M_1 \times \M_2$ 
  with unreachable states removed.
  Then $\sema{\aut{U}_1} \subseteq \sema{\aut{U}_2}$ iff for
  every state $(q_1,q_2)$ of $\M$ with $q_1 \in F_1$ we have
  $\sema{B(\M_1,F_1,q_1)} \subseteq \sema{\aut{U}_2^{q_2}}$.
\end{lemma}

\begin{proof}
  Suppose that for some state $(q_1,q_2)$ of $\M$ with $q_1 \in F_1$, we have
  $w \in \sema{B(\M_1,F_1,q_1)} \setminus \sema{\aut{U}_2^{q_2}}$.
  Let $C_1$ be the set of states visited infinitely often in $B(\M_1,F_1,q_1)$
  on input $w$, and let $C_2$ be the set of states visited infinitely
  often in $\aut{U}_2^{q_2}$ on input $w$.
  Then $C_1 = F_1$ and $C_2 \not\in \mathcal{F}_2$.
  Let $u$ be a finite word such that $\M(u) = (q_1,q_2)$.
  Then $\Inf{\M_1}(uw) = C_1 = F_1$ and $\Inf{\M_2}(uw) = C_2$,
  so $uw \in \sema{\aut{U}_1} \setminus \sema{\aut{U}_2}$.

  Conversely,
  suppose that $w \in \sema{\aut{U}_1} \setminus \sema{\aut{U}_2}$.
  For $i = 1,2$ let $C_i = \Inf{\M_i}(w)$.
  Note that $C_1 = F_1$ and $C_2 \not\in \mathcal{F}_2$.
  Let $w = xw'$, where $x$ is a finite prefix of $w$ that is
  sufficiently long that
  the run of $\M_1$ on $w$ does not visit any state outside $C_1$ after
  $x$ has been processed, and for $i = 1,2$ let $q_i = \M_i(x)$.
  Then $(q_1,q_2)$ is a (reachable) state of $\M$, $q_1 \in F_1$,
  and the $\omega$-word
  $w'$, when processed by $\M_1$ starting at state $q_1$ visits only
  states of $C_1 = F_1$ and visits each of them infinitely many times,
  that is, $w' \in \sema{B(\M_1,F_1,q_1)}$.
  Moreover, when $w'$ is processed by $\M_2$ starting at state $q_2$, the
  set of states visited infinitely often is $C_2$, 
  which is not in $\mathcal{F}_2$.
  Thus, $w' \in \sema{B(\M_1,F_1,q_1)} \setminus \sema{\aut{U}_2^{q_2}}$.
\end{proof}

To turn this into an algorithm to test inclusion for two DMAs,
$\aut{U}_1$ with automaton $\M_1$ and a single final state set $F_1$ that is
a SCC of $\M_1$ and
 $\aut{U}_2$ with automaton $\M_2$,
we proceed as follows.
Construct the product automaton 
$\M = \M_1 \times \M_2$ with unreachable states removed,
and for each state $(q_1,q_2)$ of $\M$, if $q_1 \in F_1$,
construct the DBA $B(\M_1,F_1,q_1)$ and the DMA $\aut{U}_2^{q_2}$ and
test the inclusion of language accepted by the DBA in the language
accepted by the DMA.
If all of these tests return ``yes'', 
then the algorithm returns ``yes'' for
the inclusion question for $\aut{U}_1$ and $\aut{U}_2$.
Otherwise, for the first test 
that returns ``no'' and a witness $u(v)^\omega$,
the algorithm finds by breadth-first search a minimum length
finite word $u'$ such that $\M(u') = (q_1,q_2)$, 
and returns ``no'' and the witness $u'u(v)^\omega$.

Combining this with Claim~\ref{claim:arbitrary-to-one-final-state-set}, 
we have the following.

\begin{theorem}
  \label{DMA-DMA-reduced-to-DBA-DMA}
  Let $A$ be an algorithm to test inclusion for an arbitrary DBA
  and an arbitrary DMA over the same alphabet.
  There is an algorithm to test inclusion for an arbitrary pair
  of DMAs $\aut{U}_1$ and $\aut{U}_2$ over the same alphabet
  whose running time is linear in the sizes of $\aut{U}_1$ and
  $\aut{U}_2$ plus the time for at most $k \cdot |Q_1| \cdot |Q_2|$
  calls to the procedure $A$, where $k$ is the number of final state
  sets in $\aut{U}_1$, and $Q_i$ is the state set of $\aut{U}_i$
  for $i = 1,2$.
\end{theorem}

\subsection{A DBA/DMA inclusion algorithm}

In this section, we give a polynomial time algorithm 
\proc{Incl-DBA-DMA} to test inclusion
for an arbitrary DBA and an arbitrary DMA over the same alphabet.
Assume the algorithm has inputs consisting of
a DBA $\aut{B}$ and a DMA $\aut{U}$, where
\[\aut{B} = \la \Sigma, Q_1, (q_\iota)_1, \delta_1, F \ra\]
and
\[\aut{U} = \la \Sigma, Q_2, (q_\iota)_2, \delta_2, \mathcal{F} \ra.\]

Let $\M_1$ denote the automaton of $\aut{B}$ and $\M_2$ denote the
automaton of $\aut{U}$.
The \proc{Incl-DBA-DMA} algorithm computes the product automaton 
$\M = \M_1 \times \M_2$ with unreachable states removed.
Note that any elements of $\mathcal{F}$ that are not SCCs of
$\M_2$ may be removed (in time linear in the size of $\aut{U}$)
without affecting the language accepted by $\aut{U}$.

The overall strategy of the \proc{Incl-DBA-DMA} algorithm 
is to seek a nontrivial graph theoretic strongly connected
subset $C$ of states of $G(\M)$ such that 
$\pi_1(C) \cap F \neq \emptyset$
and $\pi_2(C) \not\in \mathcal{F}$.
If such a $C$ is found, the algorithm calls 
the \proc{Witness} procedure on inputs $\M$ and $C$
to find a witness $u(v)^\omega$ such that
$\Inf{\M}(u(v)^\omega) = C$.
Because $\Inf{\M_1}(u(v)^\omega) = \pi_1(C)$ and 
$\pi_1(C) \cap F \neq \emptyset$,
$u(v)^\omega$ is accepted by $\aut{B}$.
Because $\Inf{\M_2}(u(v)^\omega) = \pi_2(C)$ 
and $\pi_2(C) \not\in \mathcal{F}$,
$u(v)^\omega$ is not accepted by $\aut{U}$.

Once the product automaton $\M$ has been computed, 
the \proc{Incl-DBA-DMA} algorithm proceeds as follows.

\paragraph{Step one.}
Compute 
the graph theoretic nontrivial strongly connected components of $G(\M)$,
say $C_1, C_2, \ldots, C_k$.
If for any $i$ with $i \in [1..k]$ we have $\pi_1(C_i) \cap F \neq \emptyset$
and $\pi_2(C_i) \not\in \mathcal{F}$, return ``no'' and the witness
returned by the \proc{Witness} procedure on inputs $\M$ and $C_i$.

\paragraph{Step two.}
Otherwise,
for each $i \in [1..k]$ such that $C_i \cap F \neq \emptyset$, 
process $C_i$ as follows.
For each element $F_j \in \mathcal{F}$ such that $F_j \subseteq \pi_2(C_i)$,
and for each state $q \in F_j$, 
compute the graph theoretic nontrivial strongly connected components,
say $D_1, D_2, \ldots, D_m$,
of the subgraph of $G(\M)$ induced by all vertices $(q_1,q_2)$ such that
$q_1 \in \pi_1(C_i)$ and $q_2 \in F_j \setminus \{q\}$.
For each such $D_s$, test whether $\pi_1(D_s) \cap F \neq \emptyset$ and
$\pi_2(D_s) \not\in \mathcal{F}$.
If so, return ``no'' and the witness returned by \proc{Witness}
on inputs $\M$ and $D_s$.

\paragraph{Step three.}
If none of the tests in Steps one or two return ``no'' and a witness,
then return ``yes''.

\begin{theorem}
  \label{DBA-DMA-inclusion}
  The \proc{Incl-DBA-DMA} algorithm solves the inclusion problem
  for an arbitrary DBA $\aut{B}$ and an arbitrary DMA $\aut{U}$
  over the same alphabet $\Sigma$.
  Any returned witness $u(v)^\omega$ has length $O((|Q_1| \cdot |Q_2|)^2)$,
  and the running time is 
  $O(|\Sigma|\cdot|Q_1|\cdot|Q_2| + |\mathcal{F}| \cdot |Q_1| \cdot |Q_2|^2)$,
  where $Q_1$ is the state set of $\aut{B}$, $Q_2$ is the state set
  of $\aut{U}$, 
  and $\mathcal{F}$ is the family of final state sets of $\aut{U}$.
\end{theorem}

\begin{proof}
  To establish the correctness of the \proc{Incl-DBA-DMA} algorithm, 
  we argue as follows.
  Suppose the returned value is ``no'' with a witness $u(v)^\omega$.
  Then the algorithm must have found
  a graph theoretic nontrivial strongly connected
  component $C$ of $G(\M)$ with $\pi_1(C) \cap F \neq \emptyset$ and
  $\pi_2(C) \not\in \mathcal{F}$ and called the \proc{Witness}
  procedure with inputs $\M$ and $C$, 
  which returned the witness $u(v)^\omega$,
  which correctly witnesses the answer ``no''.
  Thus, in this case, the returned value is correct and
  the witness has length at most $O((|Q_1| \cdot |Q_2|)^2)$.

  Suppose for the sake of contradiction that the algorithm
  \proc{Incl-DBA-DMA} returns ``yes'' but should not, that is, 
  there exists an $\omega$-word $w$ such that $w \in \sema{\aut{B}}$ and
  $w \not\in \sema{\aut{U}}$.
  Let $C$ denote $\Inf{\M}(w)$, the set of states visited
  infinitely often in the run of $\M$ on input $w$.
  Then because $w \in \sema{\aut{B}}$, $\pi_1(C) \cap F \neq \emptyset$.
  And because $w \not\in \sema{\aut{U}}$, $\pi_2(C) \not\in \mathcal{F}$.

  Clearly, $C$ is a subset of a unique $C_i$ computed in Step one, and
  $\pi_1(C_i) \cap F \neq \emptyset$.
  It must be that $\pi_2(C_i) \in \mathcal{F}$, because otherwise
  the algorithm would have returned ``no'' with the witness computed
  from $C_i$.
  Let $F_1, F_2, \ldots, F_\ell$ denote the elements of $\mathcal{F}$
  that are subsets of $\pi_2(C_i)$.

  Consider the collection
  \[S = \{F_r \mid r \in [1..\ell] \wedge \pi_2(C) \subseteq F_r\},\]
  of all the $F_r$ contained in $\pi_2(C_i)$ that contain $\pi_2(C)$.
  The collection $S$ is nonempty because $C \subseteq C_i$, and
  therefore $\pi_2(C) \subseteq \pi_2(C_i)$, and $\pi_2(C_i) \in \mathcal{F}$,
  so at least $\pi_2(C_i)$ is in $S$.
  Let $F_j$ denote a minimal element (in the subset ordering) of $S$.
  
  Then $\pi_2(C) \subseteq F_j$ but because $\pi_2(C) \not\in \mathcal{F}$,
  it must be that $\pi_2(C) \neq F_j$.
  Thus, there exists some $q \in F_j$ that is not in $\pi_2(C)$.
  When the algorithm considers this $F_j$ and $q$, then
  because $\pi_2(C) \subseteq F_j \setminus \{q\}$, $C$ is
  contained in one of the graph theoretic nontrivial
  strongly connected components $D_s$ of the subgraph
  of $G(\M)$ induced by the vertices $(q_1,q_2)$ such that
  $q_1 \in C_i$ and $q_2 \in F_j \setminus \{q\}$.

  Because $C \subseteq D_s$, and $\pi_1(C) \cap F \neq \emptyset$, we
  have $\pi_1(D_s) \cap F \neq \emptyset$.
  Also, $\pi_2(C) \subseteq \pi_2(D_s) \subseteq F_j$, but because
  $q \not\in \pi_2(D_s)$, $\pi_2(D_s)$ is a proper subset of $F_j$.
  When the algorithm considers $D_s$, 
  because $\pi_1(D_s) \cap F \neq \emptyset$,
  it must find that $\pi_2(D_s) \in \mathcal{F}$, or else it
  would have returned ``no''.
  But then $\pi_2(D_s)$ is in $S$ and is a proper subset of $F_j$, 
  contradicting our choice of $F_j$ as a minimal element of $S$.
  Thus, if all the tests in Steps one and two pass, 
  the returned value ``yes'' is correct.
  
  To analyze the running time of the \proc{Incl-DBA-DMA} algorithm,
  consider first the task of determining for a graph theoretic nontrivial
  strongly connected component $C$ of a subgraph of $G(\M)$,
  whether $\pi_1(C) \cap F \neq \emptyset$ and 
  whether $\pi_2(C) \in \mathcal{F}$.
  We assume that time linear in $|C|$ suffices for these tests, 
  using hash tables constructed in time linear in the sizes of 
  $\aut{B}$ and $\aut{U}$.

  In Step one, the computation of the graph theoretic
  strongly connected components may be carried out in time
  linear in the size of $\M$.
  Checking each resulting component $C_i$ can be done
  in time linear in $|C_i|$, so the total time is linear
  in the size of $\M$, that is $O(|\Sigma|\cdot|Q_1|\cdot|Q_2|)$.

  For Step two, for each component $C_i$ such that 
  $\pi_1(C_i) \cap F \neq \emptyset$
  there are at most $|\mathcal{F}|$ sets $F_s$
  to consider, and for each of them at most
  $|F_s|$ computations of nontrivial graph theoretic
  strongly connected components of subgraphs of $G(\M)$,
  and tests of each of the resulting components.
  Each $F_s$ has size at most $|Q_2|$, so the subgraph
  of $G(\M)$ considered has size at most $|C_i| \cdot |Q_2|$,
  and the total time in Step two is 
  $O(|\mathcal{F}| \cdot |Q_1| \cdot |Q_2|^2)$.
\end{proof}

Combining Theorem~\ref{DMA-DMA-reduced-to-DBA-DMA},
Theorem~\ref{DBA-DMA-inclusion}, 
and the reduction of equivalence to inclusion, we have the following.

\begin{corollary}
  \label{cor:DMA-inclusion-equivalence}
  There are polynomial time algorithms to solve the inclusion and
  equivalence problems for two arbitrary DMAs $\aut{U}_1$ and $\aut{U}_2$
  over the same alphabet $\Sigma$.
  If for $i = 1,2$, $Q_i$ is the set of states and $\mathcal{F}_i$
  is the family of final state sets of $\aut{U}_i$, then the
  length of any returned witness is $O((|Q_1|\cdot|Q_2|^2)$ and
  the total running time is
  \[O(|\mathcal{F}_1| \cdot |\mathcal{F}_2| \cdot |Q_1|^2 \cdot |Q_2|^3 +
  |\Sigma| \cdot |\mathcal{F}_1| \cdot |Q_1|^2 \cdot |Q_2|^2).\]
\end{corollary}

The running time bound reflects repeating
once for each element of $\mathcal{F}_1$ and each pair $(q_1,q_2)$ in $Q_1 \times Q_2$,
the cost of a call to the procedure to test inclusion for a DBA
and a DMA.

\section{Computing the right congruence automaton}

Let $\aut{A}$ be a DBA, DCA, DPA, or DMA.
Recall that $\aut{A}^q$ is the acceptor $\aut{A}$
with the initial state changed to $q$.
Then $\sema{\aut{A}^q}$ is the set of all $\omega$-words
accepted from the state $q$.
Thus, if $q_1$ and $q_2$ are two states of $\aut{A}$,
testing the equivalence of $\aut{A}^{q_1}$ to $\aut{A}^{q_2}$
determines whether these two states have the same
right congruence class, and, if not, 
returns a witness $u(v)^\omega$ that is accepted 
from exactly one of the two states.
The following is a consequence of 
Corollaries~\ref{cor:DPA-inclusion-equivalence}, 
\ref{cor:DBA-DCA-inclusion-equivalence}, and
\ref{cor:DMA-inclusion-equivalence}.

\begin{lemma}
  \label{lem:right-congruence-test}
  There is a polynomial time procedure to test 
  whether two states of an arbitrary
  DBA, DCA, DPA, or DMA $\aut{A}$ have the same right congruence class, 
  returning the answer ``yes'' if they do, and
  returning ``no'' and a witness $u(v)^{\omega}$ accepted from
  exactly one of the states if they do not.
  A returned witness $u(v)^\omega$ has length $O(|Q|^4)$,
  where $Q$ is the set of states of $\aut{A}$.
\end{lemma}

This in turn can be used in an algorithm
\proc{Right-Con} to construct the right congruence automaton 
for a given DBA (or DCA, DPA, or DMA).
Let the input acceptor be
\[\aut{A} = \la \Sigma, Q, q_\iota, \delta, \alpha \ra.\]
The algorithm constructs an automaton
\[\M = \la \Sigma, Q', \varepsilon, \delta' \ra,\]
isomorphic to the right congruence automaton of $\sema{\aut{A}}$
in which the states are represented by finite words
and the initial state is the empty word $\varepsilon$.

We describe the process of constructing $Q'$ and $\delta'$,
where $Q'$ initially contains just the empty word, and $\delta'$
is completely undefined.
While there exists a word $x \in Q'$ and a symbol $\sigma \in \Sigma$
such that $\delta'(x,\sigma)$ has not yet been defined,
loop through the words $y \in Q'$ and ask whether 
the states
$\delta(q_\iota,x\sigma)$
and $\delta(q_\iota,y)$ have the same right congruence class
in $\aut{A}$.
If so, then define $\delta'(x,\sigma)$ to be $y$.
If no such $y$ is found, then the finite word $x\sigma$ is added as
a new state to $Q'$, and the process continues.

This process must terminate because the right congruence automaton
of $\aut{A}$ cannot have more than $|Q|$ states.
When it terminates, the automaton $\M$ is isomorphic to the
right congruence automaton of $\aut{A}$.

Note that each time the equivalence algorithm returns ``no'' to a
call with $\delta(q_\iota,x \sigma)$ and $\delta(q_\iota,y)$, it
also returns a witness $u(v)^\omega$ such that exactly one of
$x\sigma u(v)^\omega$ and $yu(v)^\omega$ is accepted by $\aut{A}$.
Thus, if the \proc{Right-con} algorithm collects, 
for each new state $x \sigma$ added to $Q'$, 
the set of witnesses $u(v)^\omega$ distinguishing
it from the previous elements of $Q'$, the resulting set $D$ of
witnesses are sufficient to distinguish every pair of states of
the final automaton $\M$.

\begin{theorem}
  \label{right-con-algorithm}
  The \proc{Right-con} algorithm with input $\aut{A}$ (a DBA, DCA,
  DPA or DMA) runs in polynomial time and returns $\M$, an
  automaton isomorphic to the right congruence automaton of
  $\sema{\aut{A}}$, and $D$, a set of witnesses $u(v)^\omega$ such
  that for any $x_1, x_2 \in \Sigma^*$, if
  $\M(x_1) \neq \M(x_2)$ then there exists some $u(v)^\omega \in D$
  such that exactly one of $x_1u(v)^\omega$ and $x_2u(v)^\omega$ is
  in $\sema{\aut{A}}$.
\end{theorem}

\bibliography{df}
\bibliographystyle{plain}

\end{document}